\documentclass[a4paper,UKenglish,cleveref, autoref, thm-restate]{lipics-v2021}
\title{Faster Prefix-Sorting Algorithms for Deterministic Finite Automata}

\author{Sung-Hwan Kim}{DAIS, Ca' Foscari University of Venice, Italy}{sunghwan.kim@unive.it}{https://orcid.org/0000-0002-1117-5020}{Funded by the European Union (ERC, REGINDEX, 101039208).}

\author{Francisco Olivares}{CeBiB --- Centre for Biotechnology and Bioengineering \\
Department of Computer Science, University of Chile, Chile}{folivares@uchile.cl}{https://orcid.org/0000-0001-7881-9794}{Funded by Ph.D Scholarship 21210579, ANID, Chile.}

\author{Nicola Prezza}{DAIS, Ca' Foscari University of Venice, Italy}{nicola.prezza@unive.it}{https://orcid.org/0000-0003-3553-4953}{Funded by the European Union (ERC, REGINDEX, 101039208). Views and opinions expressed are however those of the author(s) only and do not necessarily reflect those of the European Union or the European Research Council. Neither the European Union nor the granting authority can be held responsible for them.}

\authorrunning{S.-H. Kim, F. Olivares, N. Prezza} %TODO mandatory. First: Use abbreviated first/middle names. Second (only in severe cases): Use first author plus 'et al.'

\Copyright{Sung-Hwan Kim, Francisco Olivares, Nicola Prezza} %TODO mandatory, please use full first names. LIPIcs license is "CC-BY";  http://creativecommons.org/licenses/by/3.0/

 %TODO mandatory: Please choose ACM 2012 classifications from https://dl.acm.org/ccs/ccs_flat.cfm 

 %TODO mandatory; please add comma-separated list of keywords

\category{} %optional, e.g. invited paper

\relatedversion{} %optional, e.g. full version hosted on arXiv, HAL, or other respository/website
\supplement{}%optional, e.g. related research data, source code, ... hosted on a repository like zenodo, figshare, GitHub, ...
\nolinenumbers %uncomment to disable line numbering

\EventEditors{Laurent Bulteau and Zsuzsanna Lipt\'{a}k}
\EventNoEds{2}
\EventLongTitle{34th Annual Symposium on Combinatorial Pattern Matching (CPM 2023)}
\EventShortTitle{CPM 2023}
\EventAcronym{CPM}
\EventYear{2023}
\EventDate{June 26--28, 2023}
\EventLocation{Marne-la-Vall\'{e}e, France}
\EventLogo{}
\SeriesVolume{259}
\ArticleNo{24}
\usepackage[linesnumbered,ruled,vlined]{algorithm2e}
\SetKwInput{KwInput}{Input}
\SetKwInput{KwOutput}{Output}

\begin{document}

\maketitle

\begin{abstract}

Sorting is a fundamental algorithmic pre-processing technique which often allows to represent data more compactly and, at the same time, speeds up search queries on it. In this paper, we focus on the well-studied problem of sorting and indexing string sets. Since the introduction of suffix trees in 1973, dozens of suffix sorting algorithms have been described in the literature. In 2017, these techniques were extended to sets of strings described by means of finite automata: the theory of Wheeler graphs [Gagie et al., TCS'17] introduced automata whose states can be \emph{totally}-sorted according to the co-lexicographic (co-lex in the following) order of the prefixes of words accepted by the automaton. More recently, in [Cotumaccio, Prezza, SODA'21] it was shown how to extend these ideas to arbitrary automata by means of \emph{partial} co-lex orders. This work showed that a co-lex order of minimum width (thus optimizing search query times) on deterministic finite automata (DFAs) can be computed in $O(m^2 + n^{5/2})$ time, $m$ being the number of transitions and $n$ the number of states of the input DFA.

In this paper, we exhibit new combinatorial properties of the minimum-width co-lex order of DFAs and exploit them to design faster prefix sorting algorithms. In particular, we describe two algorithms sorting arbitrary DFAs in $O(mn)$ and $O(n^2\log n)$ time, respectively, and an algorithm sorting acyclic DFAs in $O(m\log n)$ time. 
Within these running times, all algorithms compute also a smallest chain partition of the partial order (required to index the DFA). 
We present an experiment result to show that an optimized implementation of the $O(n^2\log n)$-time algorithm exhibits a nearly-linear behaviour on large deterministic pan-genomic graphs and is thus also of practical interest.
\end{abstract}

%%%%% Section 1: Introduction
\section{Introduction}

In this paper, we study the problem of indexing string sets for \emph{pattern matching} queries: pre-process a set $\mathcal L \subseteq \Sigma^*$ of strings from a finite alphabet $\Sigma$ so that later we can efficiently answer queries of the form ``is a given query pattern $P\in \Sigma^*$ substring of some string in $\mathcal L$?''. 

Clearly, an algorithmic solution to this problem requires the set $\mathcal  L$ to be representable in \emph{finite} space (even though $\mathcal  L$ itself could contain an \emph{infinite} number of strings); in this paper, we focus on string sets described by finite state automata, that is, on  regular languages. Our results build on a successful line of previous research based on the following idea: after sorting all prefixes $Pref(\mathcal L)$ of the strings in $\mathcal L$ in colexicograpic (co-lex for brevity) order\footnote{Historically, the lexicographic order of suffixes was used first; however, with finite state automata the symmetric co-lex order of $Pref(\mathcal L)$ turns out to be more natural.}, pattern matching queries translate to finding the strings in $Pref(\mathcal L)$ that are suffixed by pattern string $P$. Being $Pref(\mathcal L)$ co-lex sorted, those strings form a range in co-lex order; notice that, if the sorted $Pref(\mathcal L)$ is explicitly stored, such a range can be easily found by binary search.
Recall, however, that (due to limited available working space) we work with a 
particular representation of $\mathcal L$: a finite state automaton $\mathcal A$. This requires re-formulating the pattern matching problem on  $\mathcal A$.
It is easy to see that pattern matching queries on $\mathcal L$ translate to finding paths of  $\mathcal A$ whose labels, when concatenated, form $P$.
When using $\mathcal A$ to index $\mathcal L$, the main question becomes therefore ``how does the total co-lex order on $Pref(\mathcal L)$ map onto the states of $\mathcal A$?''. In particular cases, such a mapping yields a total order among $\mathcal A$'s states. This happens, for example, when $\mathcal A$ is a path (i.e. a string; corresponding data structures include the suffix tree \cite{weiner1973linear}, the suffix array \cite{manber:soda90,gonnet1992new}, and the FM-index \cite{ferragina2000opportunistic}), a finite set of disjoint paths (eBWT \cite{mantaci2007extension}), or a labeled arborescence (XBWT \cite{xbwt-focs}). A total order on the states of $\mathcal A$ is obtained even in particular cases where $\mathcal A$ may accept an infinite language: this is the case, for example, of de Bruijn graphs (BOSS \cite{BOSS}) and Wheeler graphs \cite{gagie:tcs17} (the latter generalize all the above classes of totally-sortable labeled graphs). 

More recently, in \cite{cotumaccio:soda21,cotumaccio:arxiv22} it was shown that in the general case (arbitrary NFAs) the total co-lex order on $Pref(\mathcal L)$ maps very naturally onto a family of \emph{partial co-lex orders} among the states of $\mathcal A$. Such a family contains only one order for any given DFA, while NFAs may admit multiple admissible co-lex orders. 
Letting $p$ be the \emph{width} of a smallest-width partial order $<_{\mathcal A}$, it was shown that pattern matching queries on $\mathcal L$ can be solved in time $\tilde O(p^2)$ per query character\footnote{The notation $\tilde O$ hides factors polylogarithmic in the size of $\mathcal A$.}. Note that this generalizes the total order case $p=1$, where indeed queries take $\tilde O(1)$ time using the aforementioned solutions (e.g. indexes on strings and labeled trees). Building the index of \cite{cotumaccio:soda21,cotumaccio:arxiv22} requires the computation of a smallest \emph{chain partition} for the co-lex order $<_{\mathcal A}$, i.e. a minimum-size partition $C_1, \dots, C_p$ of $\mathcal A$'s states such that $(C_i,<_{\mathcal A})$ is a total order for each $i=1, \dots, p$ (note that the index does not require the order $<_{\mathcal A}$ itself, just a chain partition).
Letting $n$ and $m$ be the number of states and transitions of $\mathcal A$, respectively, \cite{cotumaccio:soda21} showed how to build such a chain partition in $O(n^{5/2} + m^2)$ time in the case where $\mathcal A$ is a deterministic finite automaton (DFA). The work \cite{cotumaccio:arxiv22} presented a solution running in $\tilde O(m^2)$ time w.h.p.
In the general nondeterministic (NFA) case, the problem is known to be NP-complete\footnote{Hardness follows from hardness of the $p=1$ case \cite{DBLP:conf/esa/GibneyT19}, while membership in NP follows from the fact that the properties defining a co-lex order can be checked in polynomial time, given a candidate order.}, even though polynomial algorithms do exist for co-lex \emph{pre-orders} \cite{cotumaccio2022} (which still allow indexing and whose width is never larger than that of co-lex orders).

\subsection{Our results}

In this work, we focus on the problem of computing the smallest-width partial co-lex order $<_{\mathcal A}$ when the input is a DFA. On DFAs, $<_{\mathcal A}$ has a very intuitive definition: letting $u,v$ be states of
$\mathcal A$, we have $u <_{\mathcal A} v$ if and only if $\alpha < \beta$ for every $\alpha \in I_u$ and $\beta \in I_v$, where $<$ denotes the
co-lex order among strings and $I_u$ denotes the set of strings (in fact, a regular language) labeling all paths from the source of $\mathcal A$ to $u$.
We first observe that $<_{\mathcal A}$ is completely specified by 
pairs $(\inf I_u,\sup I_u)$ over the co-lex sorted $Pref(\mathcal L)$: in fact, we prove that $u <_{\mathcal A} v$ holds if and only if $\sup I_u \leq \inf I_v$.
This allows finding a smallest chain decomposition of $<_{\mathcal A}$ in $O(n)$ time through a solution of the \emph{interval partitioning problem}, given that the co-lex ranks of strings $\inf I_u$ and $\sup I_u$ are known for each state $u$. 
Observing that these strings can be easily encoded with two pruned versions of the DFA $\mathcal A$, this leaves the problem of computing and sorting them --- ideally, in $\tilde O(m)$ time.
We give three different solutions for this problem, which could be of independent interest. The first two solutions work on arbitrary DFAs and run in time $O(mn)$ and $O(n^2\log n)$, respectively. The latter of these two solutions is based on \emph{suffix doubling}, the technique at the core of the first suffix array construction algorithm \cite{manber:soda90}, and is close to optimal on dense graphs. 
We show that an optimized implementation of this algorithm exhibits a sub-quadratic behaviour on large deterministic pan-genomic graphs (in fact, we experimentally observe a linearithmic running time). 
The third solution works on acyclic DFAs, runs in $O(m\log n)$ time,
and  generalizes a well-known algorithm for building the Burrows-Wheeler transform in an online fashion; in our case, we process the automaton's states in any topological order and, for each processed state $u$, compute $\inf I_u$ and $\sup I_u$ using the results computed on the already-processed states.

%%%%% Section 2: Preliminaries
\section{Preliminaries}

Notation $[i,j]$, where $i,j\in \mathbb N$, denotes the integer set $\{i,i+1, \dots, j\}$ (if $i>j$, then $[i,j] = \emptyset$).
Let $\Sigma$ be a finite alphabet. A \emph{finite string} $\alpha\in \Sigma^*$ (or \emph{string of finite length}) is a finite concatenation of characters from $\Sigma$. 
The notation $|\alpha|$ indicates the length of the string $\alpha$.
The symbol $\epsilon$ denotes the empty string. 
The notation $\alpha[i]$ denotes the $i$-th character from the beginning of $\alpha$; indices start from 1, so $\alpha[1]$ is the first character of $\alpha$.
Letting $\alpha,\beta\in\Sigma^*$, $\alpha\cdot \beta$ (or simply $\alpha\beta$) denotes the concatenation of strings.
The notation $\alpha[i..j]$ denotes $\alpha[i]\cdot \alpha[i+1]\cdot\ \dots\ \cdot \alpha[j]$; if $i>j$, then $\alpha[i..j]$ is the empty string $\epsilon$. 
The notation $\alpha \sqsubseteq \beta$, where $\alpha,\beta\in\Sigma^*$, indicates that $\alpha$ is a prefix of $\beta$, i.e. $\alpha = \beta[1..i]$ for some $i\leq |\beta|$.
An \emph{$\omega$-string} $\beta \in \Sigma^\omega$ (or \emph{infinite string} / \emph{string of infinite length}) is an infinite numerable concatenation of characters from $\Sigma$. 
In this paper, we work with \emph{left-infinite} $\omega$-strings, meaning that $\beta \in \Sigma^\omega$ is constructed from the empty string $\epsilon$ by prepending an infinite number of characters to it. In particular, the operation of appending a character $a\in\Sigma$ at the end of a $\omega$-string $\alpha \in \Sigma^\omega$ is well-defined and yields the $\omega$-string $\alpha a$. 
The notation $\alpha^\omega$, where $\alpha \in \Sigma^*$, denotes the concatenation of an infinite (numerable) number of copies of string $\alpha$.

\begin{definition}
    A Deterministic Finite-State Automaton (DFA) is a quintuple  $\mathcal A = (Q, \Sigma, \delta, s, F) $ where $Q$ is the finite set of states, $\Sigma$ is a finite alphabet, $\delta : Q \times \Sigma \rightarrow Q$ is the transition function, $s\in Q$ is the initial state, and $F \subseteq Q$ is the set of final states. 
\end{definition}

As is customary, we extend the transition function to words $\alpha\in \Sigma^*$ as follows: for $a\in \Sigma$, $\alpha\in \Sigma^*$, and $q\in Q$: $\delta(q,a\cdot \alpha) = \delta(\delta(q,a),\alpha)$ and $\delta(q,\epsilon) = q$.
By $\delta^{-1}(u)$, we denote the set of states from which there exists a transition to $u$: i.e. $\delta^{-1}(u)=\{v\in Q : (\exists a\in\Sigma)(\delta(v,a)=u)\}$.

In the rest of the paper, $n= |Q|$  denotes the number of states and $m = |\delta| = |\{(u,v,a)\in Q\times Q\times \Sigma :\delta(u,a)=v\}|$ the number of transitions of the DFA under consideration. 

Following \cite{alanko:soda20}, we use the following notation for the set of words reaching a given state:

\begin{definition}\label{def:I_q}
Let $\mathcal A = (Q, \Sigma, \delta, s, F) $ be a DFA.
If $q\in Q$, let $I_q$ be the set of words \emph{reaching} $q$ from the initial state: 
$$
		I_q =\{\alpha \in \Sigma^* : q = \delta(s, \alpha)\};
$$
$I_q$ is also called \emph{the regular language recognized by $q$}.
\end{definition}

The \emph{language $\mathcal L(\mathcal A)$ recognized by $\mathcal A$} is defined as $\mathcal L(\mathcal A) = \cup_{q\in F}I_q$.

The co-lexicographic (or co-lex) order of two strings $\alpha,\beta\in \Sigma^* \cup \Sigma^\omega$ is defined as follows. (i) $\epsilon < \alpha$ for every $\alpha \in \Sigma^+ \cup \Sigma^\omega$, and (ii) if $\alpha = \alpha'a$ and $\beta=\beta'b$ (with $a,b\in \Sigma$ and $\alpha',\beta'\in \Sigma^* \cup \Sigma^\omega$), $\alpha < \beta$ holds if and only if $(a<b)  \vee (a=b \wedge \alpha' < \beta')$.
In this paper, the symbols $<$ and $\le$ will be used to denote the total order between the alphabet's characters, the co-lexicographic order between strings/$\omega$-strings, and the co-lex partial order among the states of an automaton (Definition \ref{def:colex}).
The meaning of symbols $<$ and $\le$ will always be clear from the context.
In all cases, the symbol $\le$ has the following meaning: $x\le y$ if and only if $x<y$ or $x=y$ (i.e. $x<y$ or $x$ and $y$ are the same state, the same character, or the same string, depending on the context).

Let $\mathcal A = (Q, \Sigma, \delta, s, F) $ be a DFA.
We assume that $s$ has no incoming edges; any automaton can always be transformed into an equivalent automaton with this property. 
We also assume that every state is reachable from the source: for every $v\in Q$, there exists $\alpha\in\Sigma^*$ such that $\delta(s,\alpha)=v$.
Moreover, we assume \emph{input consistency}: for every $u,v,v'\in Q$ and $c,c'\in\Sigma$, if $\delta(v,c)=\delta(v',c')=u$, then $c=c'$. We denote with $\lambda(v)$ such a uniquely-defined character and take $\lambda(s) = \#$ for the source $s$, where $\# \notin \Sigma$ is such that $\#<c$ for every $c\in\Sigma$. Note that input consistency is equivalent to working with state-labeled automata. Also this assumption is not too restrictive, since any automaton can be converted into an equivalent input-consistent automaton by just multiplying its size by a factor of $|\Sigma|$.

The following concepts can be defined more in general for NFAs (see \cite{cotumaccio:soda21}), but for the purposes of this article it will be sufficient to introduce them just on DFAs:

\begin{definition}\label{def:colex}
	Let $\mathcal A = (Q, \Sigma, \delta, s, F) $ be a DFA. A \emph{ co-lex order} on $ \mathcal A$ is a partial order $ \le $ on $ Q $ that satisfies the following two axioms:
	\begin{enumerate}
		\item (Axiom 1) For every $ u, v \in Q $, if $ u < v $, then $ \lambda(u) \leq \lambda(v) $;
		\item (Axiom 2) For every $ a \in \Sigma $ and $ u, v, u', v' \in Q $, if $ u = \delta (u', a) $, $ v = \delta (v', a) $ and $ u < v $, then $ u' \le v' $.
	\end{enumerate}
\end{definition}

The \emph{width} of a partial order is the size of its largest antichain or, equivalently by Dilworth's theorem \cite{dilworth1987decomposition}, the size of a smallest chain partition of the order. 

\begin{definition} \label{def:width} The \emph{co-lex width}   of a DFA $\mathcal A$    is the  minimum  width of a co-lex order on $\mathcal A$: 
	$$\text{width}(\mathcal A)=\min \{\text{width}(\leq):\ \leq  \text{is a co-lex order on $\mathcal A$}\}$$
\end{definition}

On DFAs, the following co-lex order is of particular interest:

\begin{definition}\label{def:prec_DFA} Let $\mathcal A$ be a DFA. The relation $<_{\mathcal A}$ over $Q$  is  defined by:
	$$u <_{\mathcal A} v \text{ if and only if } (\forall \alpha \in I_u) ( \forall \beta \in I_v) ~(\alpha < \beta).    $$
\end{definition}

In fact, by \cite[Lem. 1]{cotumaccio:arxiv22} the following holds:

\begin{lemma}\label{lem:maximumorderDFAs}
	If  $\mathcal A$ is a DFA, then $<_{\mathcal A}$ is a co-lex order on $\mathcal A$ and $\text{width}(<_{\mathcal A}) = \text{width}(\mathcal A)$. The order $<_{\mathcal A}$ is called the \emph{maximum} co-lex order on $\mathcal A$.
\end{lemma}

Computing the smallest-width co-lex order is of interest because, as shown in \cite{cotumaccio:arxiv22,cotumaccio:soda21}, there exists a linear-space index over any DFA $\mathcal A$ answering \emph{subpath queries} (find all the states of $\mathcal A$ reached by a path labeled with a given query string $P$) in time proportional to $\text{width}(\mathcal A)^2$ time per query character. In fact, the index is even compressed and uses $\log(\text{width}(\mathcal A)) + \log|\Sigma| + O(1)$ bits per transition of $\mathcal A$.
Building such an index requires computing a smallest-size chain partition of $<_{\mathcal A}$. State-of-the art algorithms for this problem run in time $O(m^2 + n^{5/2})$ \cite{cotumaccio:soda21} and $\tilde O(m^2)$ w.h.p. \cite{cotumaccio:arxiv22}. The goal of our paper is to improve these bounds by exploiting a new characterization for $<_{\mathcal A}$, introduced in the next section. 

%%%%% Section 3
\section{A new characterization of  the maximum co-lex order of a DFA}\label{sec:new characterization}

In this section, we give a new 
 interval-based characterization of the maximum co-lex order $<_{\mathcal A}$ of a DFA. We show that this yields an $O(n)$-space representation for $<_{\mathcal A}$ (observe that, in general, a partial order requires $O(n^2)$ space to be represented) and that, given this representation, one can compute a smallest chain partition of $<_{\mathcal A}$ in linear $O(n)$ time. 
 
\subsection{Infimum and supremum strings}

Let $u$ be a state of a DFA $\mathcal{A}=(Q,\Sigma,\delta,s,F)$. For the set $I_u$ of strings recognized by $u\in Q$, consider a (possibly infinite) string $\beta$ such that $\beta$ is a lower bound of $I_u$; i.e. $\beta\le\alpha$ for every $\alpha\in I_u$. Consider the co-lex largest string $\gamma$ among such lower bounds of $I_u$. We call such a string the $\emph{infimum string}$ of $u$, and denote it by $\inf I_u$. Similarly, we define the supremum string $\sup I_u$ of $u$ as the least upper bound of $I_u$; see Figure \ref{fig:dfa} for an example.

\begin{definition}[Infimum and supremum strings]
Let $u\in Q$ be a state of a DFA $\mathcal{A}=(Q,\Sigma,\delta,s,F)$. The infimum string $\inf I_u$ and the supremum string $\sup I_u$ are defined as:
\begin{align*}
\inf I_u &= \gamma\in\Sigma^*\cup\Sigma^\omega \mbox{ s.t. } (\forall \beta\in\Sigma^*\cup\Sigma^\omega \mbox{ s.t. } (\forall\alpha\in I_u~\beta\le\alpha)~ \beta\le\gamma) \\
\sup I_u &= \gamma\in\Sigma^*\cup\Sigma^\omega \mbox{ s.t. } (\forall \beta\in\Sigma^*\cup\Sigma^\omega \mbox{ s.t. } (\forall\alpha\in I_u~\alpha\le\beta)~ \gamma\le\beta) 
\end{align*}
\end{definition}

\begin{figure}[th!]
\centering
\subfloat{
\includegraphics[page=2,width=60mm]{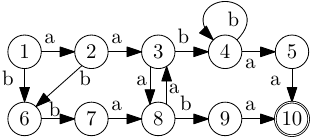}
}
\subfloat{
\begin{tabular}{crr|crr}\hline
$i$ & $\inf I_{v_i}$ & $\sup I_{v_i}$ & 
$i$ & $\inf I_{v_i}$ & $\sup I_{v_i}$ \\\hline 
1 & $\epsilon$ & $\epsilon$ &
6 & b & ab \\
2 & a & a & 
7 & bb & abb \\
3 & aa & abbaa &
8 & aaa & abba \\
4 & aab & b$^\omega$ &
9 & aaab & abbab \\
5 & aaba & b$^\omega$a &
10 & aabaa & abbaba 
\\\hline
\end{tabular}
}
\caption{Example DFA with its infimum/supremum strings.
}
\label{fig:dfa}
\end{figure}

As a warm up, we make several observations on $I_u$, infimum, and supremum strings.
\begin{observation} \label{observation:minmax_string}
Let $\mathcal A=(Q,\Sigma,\delta,s,F)$ be a DFA. For any $u\in Q$, the following hold:

\begin{enumerate}
\item For every $\alpha\in I_u$, $\alpha$ is finite. \label{obs:2}
\item For any $v(\neq u)\in Q$, $I_u\cap I_v$ is the empty set. \label{obs:3:empty:intersection:IuIv}
\item For any finite suffix $\alpha\in\Sigma^*$ of $\inf I_u$ (or $\sup I_u$), there exists $\beta\in\Sigma^*$ such that $\beta \alpha\in I_u$. \label{obs:4:sufofIu:exist:betaalphaIu}
\item $\inf I_u\in I_u$ if and only if $\inf I_u$ is a finite string; similar for $\sup I_u$. \label{obs:7:finite:Iu}
\item $I_u$ 
is a singleton if and only if $\inf I_u=\sup I_u$. In such a case, $I_u=\{\inf I_u(=\sup I_u)\}$ and $\inf I_u=\sup I_u \in I_u$ is a finite string. \label{obs:8:singleton:Iu}
\item For $v(\neq u)\in Q$, if $\inf I_u=\inf I_v$ or $\inf I_u=\sup I_v$ then $\inf I_u$  has infinite length; similar for $\sup I_u$. \label{obs:9:same:infsuf:for:u:v:infinite}
\end{enumerate}
\end{observation}

\begin{proof}
\begin{enumerate}
\item By definition of $I_u$.
\item By definition of DFA, 
for any string $\alpha$, there exists only one state $u$ such that $\delta(s,\alpha) = u$.

\item
Let $\alpha$ be any finite suffix of $\inf I_u$; the $\sup I_u$ case is analogous. We claim that there must exist $v\in Q$ such that $\delta(v,\alpha)=u$. This will prove our main claim, since for any string $\beta\in I_v$, $\delta(s,\beta\cdot\alpha)=\delta(\delta(s,\beta),\alpha))=\delta(v,\alpha)=u$ by definition, so $\beta\cdot\alpha \in I_u$. 

To prove the claim, assume by contradiction that there is no such $v\in Q$. Let $\alpha'\in\Sigma^*$ be the longest suffix of $\alpha$ such that there exists $v'\in Q$ such that $\delta(v',\alpha')=u$. Let $\alpha''\in\Sigma^*\cup\Sigma^\omega$ and $a\in\Sigma$ be a string and a symbol such that $\inf I_u=\alpha''\cdot a\cdot \alpha'$.
Let $b\in\Sigma$ be the smallest alphabet symbol that is greater than $a$.
Note that such $b$ must exist; if not, every $v'$ must have an incoming transition labeled by a symbol $a'(\in\Sigma)\le a$; since $a'\neq a$ (otherwise $\alpha'$ would be longer), there exists a string in $I_u$ suffixed by $a'\alpha'$ which is co-lex smaller than $\inf I_u$, which causes a contradiction.
Then, for every such $v'$ and every $v''\in Q$ and $c\in\Sigma$ such that $\delta(v'',c)=v'$, we have $a<b\le c$. Note that $\inf I_u<b\cdot\alpha'\le \gamma$ for every $\gamma\in I_u$, so $\inf I_u$ is not the greatest lower bound of $I_u$. This is a contradiction with the definition of $\inf I_u$.

\item $(\Rightarrow)$ By definition of $I_u$.

$(\Leftarrow)$ 
Let $\inf I_u$ be finite. Let us assume, by contradiction, that $\inf I_u\notin I_u$. Let $\alpha=a\cdot \inf I_u$ where $a\in\Sigma$ is the smallest symbol of the alphabet. We claim that $(\inf I_u<)\alpha\le\beta$ for every $\beta\in I_u$, which contradicts the definition of $\inf I_u$. Consider a $\beta\in I_u$. Let $k$ be the length of the longest common suffix of $\beta$ and $\inf I_u$. If $k<|\inf I_u|$, then obviously $\alpha<\beta$ because prepending a symbol to $\inf I_u$ does not affect the relative co-lex order of $\inf I_u$ and $\beta$. If $k=|\inf I_u|$, then $\inf I_u$ is a suffix of $\beta$ and $|\inf I_u|+1\le |\beta|$ because $\inf I_u\notin I_u$. Therefore after prepending the smallest symbol $a$ to $\inf I_u$, we still have $a \cdot \inf I_u \le \beta$.

To prove the other case for $\sup I_u$, let $\sup I_u$ be finite and let us assume for a contradiction that $\sup I_u\notin I_u$. From (\ref{obs:4:sufofIu:exist:betaalphaIu}), there exists $\beta\in\Sigma^*$ such that $\beta\alpha\in I_u$ where $\alpha=\sup I_u$; note that a string is a suffix of itself. Because $\delta(s,\alpha)\neq u$, it holds $\beta\neq\epsilon$. However, then we have $\sup I_u=\alpha<\beta\alpha\in I_u$, which contradicts with the definition of $\sup I_u$ being an upper bound of $I_u$.

\item $(\Rightarrow)$
if $I_u = \{\alpha\}$, then clearly $\inf I_u = \alpha$ and $\sup I_u = \alpha$, so $\inf I_u=\sup I_u$.
$(\Leftarrow)$ If $\inf I_u=\sup I_u$, then they are the same \emph{finite} string. To see this, assume by contradiction that $\inf I_u=\sup I_u$ have infinite length. Then, for every $\alpha\in I_u$, $\inf I_u<\alpha<\sup I_u$. Since $\inf I_u=\sup I_u$, no such $\alpha$ can exist thus $I_u = \emptyset$. This is a contradiction, because  it must hold $I_u\neq\emptyset$ by the assumption that there always exists $\alpha\in\Sigma^*$ such that $\delta(s,\alpha)=u$ (and, for the source $s$, $I_s = \{\epsilon\}$).
Since $\inf I_u(=\sup I_u)$ is a finite string, $\inf I_u\in I_u$ by (\ref{obs:7:finite:Iu}). In addition, $\inf I_u$ is the unique string in $I_u$ because for every $\alpha\in I_u$, $\inf I_u\le \alpha\le\sup I_u$ and $\inf I_u=\sup I_u$, therefore $\inf I_u=\alpha=\sup I_u$.
\item Immediate from Observations (\ref{obs:3:empty:intersection:IuIv}) and (\ref{obs:7:finite:Iu}).
\end{enumerate}
\end{proof}

To conclude the section, we prove a lemma showing that infimum and supremum strings can always be expressed as a (possibly, infinite) concatenation of a constant number of distinct strings whose length does not exceed the number of states.
This lemma will be useful later to bound the sorting depth of our algorithms computing $<_{\mathcal A}$.

\begin{lemma}\label{lem:length inf/sup}
For a DFA $\mathcal{A}=(Q,\Sigma,\delta,s,F)$ and a state $u\in Q$, let $\gamma\in\{\inf I_u,\sup I_u\}$ be either the infimum or the supremum string of $I_u$. Then,
\begin{enumerate}
\item If $\gamma$ is finite, then $|\gamma|<|Q|$.
\item If $\gamma$ has infinite length, then $\gamma=\beta^\omega\alpha$ for some $\alpha,\beta\in \Sigma^*$ such that $|\alpha|+|\beta|<|Q|$.
\end{enumerate}
\label{lemma:infsup:string:form}
\end{lemma}
\begin{proof} 
Suppose $|\gamma|\ge |Q|$. 
Let $\gamma = \gamma'\gamma''$, where $\gamma''$ is the length-$|Q|$ suffix of $\gamma$.
Consider a sequence of states $v_1,v_2,\cdots,v_{|Q|+1}$ such that 
$v_{|Q|+1}=u$ and $\delta(v_{k},\gamma''[k])=v_{k+1}$ for $1\le k\le |Q|$.
Note that $v_i\neq s$ for every $2\le i\le|Q|+1$ because the start state $s$ does not have an incoming transition. Then, there are at most $|Q|-1$ distinct states among the $|Q|$ states $v_2,\cdots,v_{|Q|+1}$ so by the pigeonhole principle there must be $2\le i<j\le |Q|+1$ such that $v_i=v_j$. 
Let $\beta' = \gamma'  \gamma''[1..i-1]$, $\beta = \gamma''[i..j-1]$, and $\alpha=\gamma''[j..|Q|]$. 
Note that $\beta \alpha$ is a proper suffix of $\gamma''$ (\emph{proper} because $i\geq 2$), therefore $|\alpha| + |\beta| =  |\beta\alpha| < |\gamma''| = |Q|$. Note also that (by definition of $\beta', \beta$, and $\alpha$) $\gamma = \beta'\beta\alpha$.

Let us assume $\gamma=\inf I_u$; the other case $\gamma=\sup I_u$ is analogous.
Note that $\gamma=\beta'\beta\alpha\le \beta'\beta^k\alpha$ for every $k\ge 0$. To see this, observe that if $\gamma$ is a finite string, then $\gamma=\beta'\beta\alpha\in I_u$ by Observation \ref{observation:minmax_string}.\ref{obs:7:finite:Iu}. Since $\delta(s,\beta')=v_i$, $\delta(v_i,\beta)=v_j=v_i$, and $\delta(v_j(=v_i),\alpha)=v_{|Q|+1}=u$, we have $\beta'\beta^k\alpha\in I_u$ for every $k\ge 0$. By definition of $\inf I_u$, $\gamma=\inf I_u\le \beta'\beta^k\alpha$ for every $k\ge 0$.
On the other hand, if $\gamma$ has infinite length, assume for a contradiction that there exists $k'\ge 0$ such that $\beta'\beta^k\alpha<\beta'\beta\alpha=\gamma$. Consider the length-$(l+1)$ suffix $\beta''\alpha$ of $\beta'\beta^k\alpha$ where $l$ is the length of the longest common suffix between $\beta'\beta^k\alpha$ and $\gamma$; clearly, such a $l$ is finite and $l \geq |\alpha|$ since $\alpha$ suffixes both strings. Then by Observation \ref{observation:minmax_string}.\ref{obs:4:sufofIu:exist:betaalphaIu}, there exists $\beta'''\in\Sigma^*$ such that $\beta'''\beta''\alpha\in I_u$. However we have $\beta'''\beta''\alpha<\beta'\beta\alpha=\gamma=\inf I_u$, which contradicts the definition of $\inf I_u$.

By plugging $k=0$ into the inequality $\beta'\beta\alpha \leq \beta'\beta^k\alpha$ above, we obtain $\beta'\beta\alpha \leq \beta'\alpha$. Equivalently (by removing the common suffix $\alpha$) it holds $\beta'\beta \leq \beta'$; but then, we can plug again a common suffix $\beta^k\alpha$ for any $k\geq 0$ and obtain that $\beta'\beta^{k+1}\alpha \leq \beta'\beta^k\alpha$ for any $k\geq 0$. In particular, this implies that $\beta'\beta^k\alpha \leq \beta'\beta\alpha = \gamma$ for any $k\geq 1$.

Since in the previous two paragraphs we proved that $\gamma\le \beta'\beta^k\alpha$ and $\beta'\beta^k\alpha\le\gamma$ for any  $k\ge 1$, we conclude that $\gamma = \beta'\beta^k\alpha$ for any $k\geq 1$, i.e. $\gamma$ must be an $\omega$-string of the form $\gamma = \beta^\omega\alpha$. 
This proves claim (2). Claim (1) also follows since the assumption that $\gamma$ is finite and $|\gamma| \geq |Q|$ leads to $\gamma = \beta^\omega\alpha$ (a contradiction to the finiteness of $\gamma$), hence its negation (i.e. claim 1) must hold.
\end{proof}

\subsection{\texorpdfstring{$O(n)$}{O(n)}-space representation of \texorpdfstring{$<_\mathcal{A}$}{<\_A}}\label{sec:interval representation}

Let $\mathcal K(\mathcal A) = \{\inf I_u : u\in Q\} \cup \{\sup I_u : u\in Q\} \subseteq \Sigma^* \cup \Sigma^\omega$ be the set of all infimum and supremum strings of $\mathcal A$. Let $rank(\alpha)$, for  $\alpha \in \mathcal K(\mathcal A)$, denote the position of $\alpha$ in the total order $(\mathcal K(\mathcal A), <)$ (e.g. $rank(\alpha) = 1$ for the co-lex smallest string $\alpha \in \mathcal K(\mathcal A)$, and so on).

Our new representation of $<_{\mathcal A}$ is the set of $n$ integer pairs $\{(rank(\inf{I_u}), rank(\sup{I_u}))\ :\ u\in Q\} \subseteq [1,2n]\times [1,2n]$ (note that $|\mathcal K(\mathcal A)| \leq 2n$). With the next theorem, we show that this set is indeed sufficient to reconstruct $<_{\mathcal A}$.

\begin{theorem}\label{thm:colex_characterization}
    Let $\mathcal A = (Q, \Sigma, \delta, s, F) $ be a DFA. Then, for any $u,v(\neq u)\in Q$, $u <_{\mathcal A} v$ if and only if $\sup{I_u}\le \inf{I_v}$.
 \end{theorem}

\begin{proof}

$(\Rightarrow)$
To prove $u<_\mathcal{A} v \Rightarrow \sup I_u \leq \inf I_v$ for all $u,v\in Q$, assume by contradiction that there exist $u,v\in Q$ such that $u<_\mathcal{A} v$ and $\inf I_v < \sup I_u$.
We claim that, in this case, there must exist $\alpha\in I_u$, $\beta\in I_v$ such that $\beta<\alpha$. By Definition \ref{def:prec_DFA}, this contradicts $u<_\mathcal{A} v$.
First, note that there must exist $\alpha\in I_u$ such that $\inf I_v<\alpha$, otherwise it would be $\sup I_u\le \inf I_v$. 
We divide the proof by contradiction in the two  cases (i) $\inf I_v$ is a finite string and (ii) $\inf I_v$ has infinite length. 

(i) If $\inf I_v$ is finite, then $\inf I_v\in I_v$ by Observation \ref{observation:minmax_string}.\ref{obs:7:finite:Iu}. Choosing $\beta=\inf I_v$, we have  $\beta=\inf I_v (\in I_v) <\alpha$. This contradicts $u<_\mathcal{A} v$.

(ii) If $\inf I_v$ has infinite length, then by Lemma \ref{lemma:infsup:string:form} we can write it as $\inf I_v=\gamma_2^\omega \gamma_1$ for some strings $\gamma_1,\gamma_2\in \Sigma^*$. Note that, for every $k\ge 0$, there exists a string $\gamma_3\in\Sigma^*$ such that $\gamma_3\gamma_2^k\gamma_1\in I_v$ (by Observation \ref{observation:minmax_string}.\ref{obs:4:sufofIu:exist:betaalphaIu} because $\gamma_2^k\gamma_1$ is a suffix of $\inf I_u$).
Choose any integer $k'$  such that $|\gamma_2^{k'}\gamma_1|>|\alpha|$ (such an integer exists since $\alpha$ is finite). 
Since $\inf I_v=\gamma_2^\omega\gamma_1<\alpha$, we also have $\gamma_3\gamma_2^{k'}\gamma_1(=\beta\in I_v) <\alpha$. Again, this contradicts $u<_\mathcal{A} v$.

$(\Leftarrow)$ Let $\sup{I_u}\le \inf{I_v}$, and choose any $\alpha\in I_u$ and $\beta\in I_v$. We need to prove that $\alpha <_{\mathcal A} \beta$.
By definition of $\sup{I_u}$ and $\inf{I_v}$, we have $\alpha \le \sup I_u \le \inf I_v \le \beta$.
If $\sup I_u<\inf I_v$, then $\alpha<\beta$.
If, on the other hand, $\sup I_u=\inf I_v$ then both $\sup I_u$ and $\inf I_v$ must be infinite strings by Observation \ref{observation:minmax_string}.\ref{obs:9:same:infsuf:for:u:v:infinite}. Since $\alpha$ and $\beta$ are both finite, it must be the case that  $\alpha\neq \sup I_u$ and $\beta \neq \inf I_v$, therefore $\alpha<\sup I_u=\inf I_v<\beta$. Since this holds for any $\alpha\in I_u$ and $\beta\in I_v$, by definition of $<_{\mathcal A}$ it holds $\alpha <_{\mathcal A} \beta$.
\end{proof}

Equivalently, Theorem \ref{thm:colex_characterization} shows that $<_{\mathcal A}$ can be interpreted as a set of intervals on the co-lex sorted $Pref(\mathcal L(\mathcal A))$.
This characterization of $<_{\mathcal A}$  will allow us to compute this order faster than the state-of-the-art by (i) co-lex sorting the infimum and supremum strings (Section \ref{sec:sorting}), and (ii) computing a smallest chain partition for $<_{\mathcal A}$ in linear time (Section \ref{sec:chain partition}).

\subsection{Linear-time chain partitioning algorithm}\label{sec:chain partition}

In general, a partial order over $n$ elements requires $O(n^2)$ space to be represented. Moreover, the fastest general-purpose algorithms for computing the smallest chain partition of a partial order run either in worst-case time $O(n^{5/2})$ (see, for example, \cite[Lem. 6.1]{cotumaccio:soda21}) or in $\tilde O(n^2)$ time w.h.p. \cite{KoganParterICALP}. In this section we show that given the $O(n)$-space  representation $S = \{(rank(\inf{I_u}), rank(\sup{I_u})):u\in Q\}$ of $<_{\mathcal A}$, from which the order can be represented using intervals, we can compute a smallest chain partition of this order in optimal $O(n)$ time. It is known that the optimal solution of a smallest chain partition of interval orders can be computed with a greedy method (see \cite[Sec. 6.8]{kellertrotter17}). Moreover, given the sorted intervals, one can compute it in linear time \cite{intervalcoloring}; for completeness here we give the details.

Based on Theorem \ref{thm:colex_characterization}, we now show a simple linear-time reduction from the smallest chain partition problem (where the input order is represented as described in Section \ref{sec:interval representation}) to the following problem:

\begin{definition}[\textit{Interval partitioning problem, cf. \cite [Sec.~4.1]{kleinberg:ad06}}]\label{def:interval-partitioning}
    Let $\{[s_1, f_1], \dots, [s_n, f_n]\}$ be a set of $n$ activities that must be served (each) by a device.
    One device can handle at most one activity at the same time. 
    $[s_i, f_i]$ is an interval, where $s_i$ and $f_i$ are the \textit{starting} and \textit{finishing} time of activity $i$, respectively. Determine the minimum number of devices to serve all the activities.
\end{definition}

Let $S = \{a_1 = (s_1, f_1), \dots, a_n = (s_n, f_n)\}$ be an instance of the \textit{smallest chain partition} problem for $<_{\mathcal A}$ (that is, a particular instance of $<_{\mathcal A}$). Our reduction from this instance to an instance of the \textit{interval partitioning problem} works as follows:

\begin{enumerate}
    \item For each pair $a_i = (s_i, f_i)$, with $i \in [1..n]$, let $s_i' = 2s_i + 1$ and $f_i' = 2f_i$.
    \item Return the set of intervals $S'' = \{a''_i\}_{i=1}^n$, where $a_i'' = [s_i'' = \min{(s_i', f_i'), f_i'' = \max{(s_i', f_i')}}]$
\end{enumerate}

The following Lemma shows that our reduction is correct:

\begin{lemma}\label{lemma:reduction}
Let $(s_i,f_i), (s_j,f_j)$ be two input pairs, with $s_i\leq s_j$ without loss of generality. Let moreover $[s''_i,f''_i], [s''_j,f''_j]$ be the intervals into which the two pairs get transformed by the above reduction. Then, $f_i \leq s_j$ if and only if $f_i'' < s_j''$ (i.e. $[s''_i,f''_i]$ and $[s''_j,f''_j]$ do not overlap).
\end{lemma}

\begin{proof}
We divide the proof into two cases: (Case 1) at least one of $s_i=f_i$ or $s_j=f_j$ holds, and (Case 2) both $s_i<f_i$ and $s_j<f_j$ hold.

(Case 1) 
First, we show that $f_i\neq s_j$. Assume that $s_i=f_i$ (the other case $s_j=f_j$ is analogous). 
Let $u$ be the state associated with the pair $(s_i,f_i)$, and $v$ be the state associated with the pair $(s_j,f_j)$.
By Observation \ref{observation:minmax_string}.\ref{obs:8:singleton:Iu}, $s_i=f_i$ implies that $\inf I_u = \sup I_u$ is a finite string, and $\inf I_u = \sup I_u \in I_u$. 
If $\inf I_v$ is an infinite string, then clearly $\sup I_u \neq \inf I_v$ (being $\sup I_u$ a finite string), i.e. $f_i\neq s_j$. If, on the other hand, $\inf I_v$ is a finite string, then by Observation
\ref{observation:minmax_string}.\ref{obs:7:finite:Iu} we have $\inf I_v \in I_v$; since by Observation \ref{observation:minmax_string}.\ref{obs:3:empty:intersection:IuIv}, we have $I_u \cap I_v = \emptyset$, also in this case we derive that $\sup I_u \neq \inf I_v$, i.e. $f_i\neq s_j$.

Knowing $f_i\neq s_j$, we obtain that $f_i\le s_j\Leftrightarrow f_i<s_j \Leftrightarrow f_i+1\le s_j$.
Note that, since $s_j''=2s_j+1>2s_j$ (if $s_j\neq f_j$) or $s_j''=2f_j=2s_j$ (if $s_j=f_j$), we have 
$2s_j\le s_j''$. Similarly, $f_i''\le 2f_i +1$. Hence $2s_i \le s_i''<f_i''\le 2f_i+1$ (note that $s_i''<f_i''$ always holds for any interval in our reduction).
Therefore, we have $f_i \le s_j \Rightarrow f_i+1\le s_j\Rightarrow f_i'' \le 2f_i+1 < 2(f_i+1) \le 2s_j \le s_j''$. 
For the other direction, note that 
$2f_i\le f_i''$ and $s_j''\le 2s_j+1$.
Then, using these inequalities we obtain: $f_i''<s_j'' \Rightarrow 2f_i\le f_i''<s_j''\le 2s_j+1\Rightarrow 2f_i<2s_j+1 \Rightarrow f_i\le s_j$. 

(Case 2) In this case, we have $f_i\le s_j\Rightarrow f_i''=2 f_i < 2 s_j+1 = s''_j \Rightarrow f_i'' < s''_j$. For the other direction, note that 
$f_i'' < s''_j \Rightarrow 2 f_i = f_i'' < s''_j = 2 s_j+1 \Rightarrow f_i \le s_j$.
\end{proof}

By Lemma \ref{lemma:reduction}, we can now solve \textit{smallest chain partition problem} for the particular order $<_{\mathcal A}$ by reducing it to an instance of the \textit{interval partitioning problem}. Moreover, it is easy to see that the reduction works in linear time so the linearity of our strategy relies on the cost of the algorithm we use to solve the latter problem. We can use a greedy method (cf. \cite{intervalcoloring,intervalcoloring2}) to optimally solve the interval partitioning problem (namely, using the smallest possible number of devices).
The algorithm processes the intervals in non-decreasing order of starting times, breaking ties arbitrarily.
For each interval, we choose any idle device among the available ones. We can keep track of the list of the available devices if the starting and finishing times of the intervals are already sorted. If all devices are busy, we add a new device.

The above-sketched algorithm spends amortized constant time on every activity, plus the time required to sort the input set of intervals. As said earlier, the elements of our input pairs (i.e. before the reduction) are integer values in the range $[1,2n]$.
After the reduction, this range gets expanded to $[2,4n+1]$. This allows us to radix-sort the intervals in $O(n)$ time. As a result, in our scenario we can solve the \textit{interval partition problem} in $O(n)$ time and, in particular, find the \textit{smallest chain partition} of $<_{\mathcal A}$ given its ranked-pair representation in linear time. 

%%%%% Section 4
\section{Co-lex sorting infimum/supremum strings}\label{sec:sorting}

In this section, we present three algorithms to compute and sort the set containing all infimum and supremum strings of a DFA.
The first two algorithms sort the strings in such a way that for every iteration the strings are co-lex sorted with respect to a longer suffix; we present one simple solution that increases the suffix length by 1 at each iteration, and one that doubles the suffix length at each iteration. The third algorithm is a generalization of online BWT construction and is based on the online algorithm for sorting Wheeler DFA presented in \cite[Sec. 3.2]{alanko:soda20}. This algorithm works only on acyclic DFAs but has a lower time complexity than the former two solutions.

For ease of explanation, we consider only infimum strings since the supremum string case is analogous. 
Indeed, one can easily compute and sort both infimum and supremum strings at the same time by creating two copies of the input DFA and then running our algorithms on the union of the two DFAs, extracting the infimum strings on one DFA and the supremum strings on the other DFA while at the same time sorting the union of these two string sets. 

\subsection{Simple \texorpdfstring{$O(mn)$}{O(mn)}-time algorithm}

Let us establish some notations before describing the algorithm.
For a (possibly infinite) string $\alpha$ and an integer $k\ge 0$, we denote by $suf_k(\alpha)$ the length-$k$ suffix of $\alpha$.
When $|\alpha|<k$, we pad $suf_k(\alpha)$ by prepending $k-|\alpha|$ copies of a special symbol $\# \notin \Sigma$, with $\# < c$ for all $c\in \Sigma$; in this way, we guarantee that $suf_k(\alpha)$ is always a string of length $k$ and we do not affect the co-lex order of such suffixes (which remains the same before and after the padding).

For state $u\in Q$ and integer $k\ge 0$, we denote by $rank_k(u) \in \mathbb N$ the intermediate rank at iteration $k$ of $u$ in the total order we are computing; this integer indicates the co-lex rank of $suf_k(\inf I_u)$ among $\{ suf_k(\inf I_v) : v\in Q\}$. More formally, for $u\in Q$ and $k\ge 0$,
\[
\begin{array}{l}
    rank_0(u) = 1 \\
    rank_k(u) = |\{ suf_k(\inf I_v): v\in Q~\land ~suf_k(\inf I_v)\le suf_k(\inf I_u) \}| \ \ \mathrm{for}\ k>0 
\end{array}
\]
Observe that two states are assigned the same rank if their corresponding length-$k$ suffixes are equal. 
The algorithm works by \emph{pruning} transitions of the input automaton, i.e. by removing, for every state $u$, transitions coming from a state with a non-minimum rank among the predecessors of $u$. We denote by $\delta_k$ the (pruned) transition function at iteration $k$.

The algorithm works as follows. 
At iteration $k\ge 0$, we perform the following operations:

\begin{enumerate}
\item \textbf{Compute $rank_{k+1}$.} Sort the states $\{u\in Q\}$ by their label $\lambda(u)$ with ties broken by $rank_k(v)$ for any $v\in \delta_k^{-1}(u)$ (the step below will guarantee that all predecessors $v$ of $u$ have the same $rank_k(v)$).
\item \textbf{Compute $\delta_{k+1}$.} For each $u\in Q$,  keep only the transitions from the min-rank predecessors: for $v\in\delta_{k}^{-1}(u)$, $v\in\delta_{k+1}^{-1}(u)$ iff $rank_{k+1}(v)=\min \{rank_{k+1}(u'):u'\in \delta_k^{-1}(u)\}$.
\end{enumerate}

As far as the running time of each iteration is concerned, computing $rank_{k+1}$ can be done in $O(n)$ time by 2-pass radix sorting (that is, by incoming label and breaking ties by any predecessor's rank $rank_{k}$). Computing $\delta_{k+1}$ takes $O(|\delta_k|)=O(|\delta|) = O(m)$ time. Hence, each iteration takes $O(m)$ time.

Since $\forall k\ge 0$, $\forall u\in Q$, and $\forall v\in \delta_k^{-1}(u)$ we have $suf_{k+1}(\inf I_u)=suf_{k}(\inf I_v)\cdot \lambda(u)$, it is easy to see that the following invariant always holds at the beginning of iteration $k$: the infimum strings are sorted with respect to the co-lex order of their length-$k$ suffixes.
This invariant shows that the number of iterations we have to perform is exactly the length of the suffixes that need to be sorted to obtain the correct co-lex order of the infimum strings.
We are left to find an upper bound to this length; observe that this is not a trivial problem, since infimum strings may have infinite length. 

Consider any two infimum strings $\alpha,\beta\in\{\inf I_u:u\in Q\}$. The upper bound above can be computed by upper-bounding the length of the longest common suffix between $\alpha$ and $\beta$.
If any of the two strings is finite, then by Lemma \ref{lemma:infsup:string:form} their longest common suffix does not exceed length $n$.
If both strings are infinite, then
by Lemma \ref{lemma:infsup:string:form} we can write them as $\alpha=\alpha_2^\omega\alpha_1$ and $\beta=\beta_2^\omega\beta_1$ and we can use the following:

\begin{lemma}[cf.~\cite{mantaci2007extension}]
For two infinite strings $\alpha=\alpha_2^\omega\alpha_1$ and $\beta=\beta_2^\omega\beta_1$, where $\alpha_1,\beta_1\in\Sigma^*$ and $\alpha_2,\beta_2\in \Sigma^+$, let $\alpha'$ and $\beta'$ be their suffixes of length $k=|\alpha_2|+|\beta_2|+\max\{|\alpha_1|,|\beta_1|\}$. Then, $\alpha' < \beta'$ if and only if $\alpha < \beta$.
\label{lemma:suffix:length}
\end{lemma}
\begin{proof}

Without loss of generality, let us assume $|\alpha_1|\le |\beta_1|$. 
Moreover, note that without loss of generality we can also assume that $|\alpha_2|+|\alpha_1| > |\beta_1|$; if this does not hold, then re-write $\alpha_1 \leftarrow \alpha_2^k\alpha_1$ for the only integer $k>0$ such that $|\alpha_1|\le |\beta_1| < |\alpha_2|+|\alpha_1|$ holds; after the transformation, $\alpha$ can still be written as $\alpha=\alpha_2^\omega\alpha_1$.

If $\alpha_1$ is not a suffix of $\beta_1$, then clearly the longest common suffix between $\alpha$ and $\beta$ is at most $|\alpha_1|$, so the claim holds. Let us assume therefore that $\alpha_1$ is a suffix of $\beta_1$, i.e. $\beta_1 = \beta_1'\alpha_1$ for some $\beta_1'\in \Sigma^*$. 
Since by assumption $|\alpha_2|+|\alpha_1| > |\beta_1|$, note that $|\alpha_2| > |\beta_1'|$.
Similarly as above, if $\beta_1'$ does not suffix $\alpha_2$ (i.e. $\beta_1 = \beta_1'\alpha_1$ does not suffix $\alpha = \alpha_2^\omega\alpha_1$) then the longest common suffix between $\alpha$ and $\beta$ is at most $|\beta_1|$ and the claim holds, so let us assume that $\beta_1'$ is a suffix of $\alpha_2$, i.e. $\alpha_2 = \alpha_2'\beta_1'$ for some $\alpha_2'\in \Sigma^*$. Since $\alpha_2^\omega = (\alpha_2'\beta_1')^\omega = (\beta_1'\alpha_2')^\omega\beta_1'$, we conclude that comparing co-lexicographically $\alpha = (\beta_1'\alpha_2')^\omega\beta_1'\alpha_1$ and $\beta = \beta_2^\omega\beta_1'\alpha_1$ reduces to comparing $(\beta_1'\alpha_2')^\omega$ and $\beta_2^\omega$. According to \cite[Proposition 5]{mantaci2007extension}, given any $\gamma_1, \gamma_2\in \Sigma^+$ it is sufficient to compare the length-$k'$ suffixes of $\gamma_1^\omega$ and $\gamma_2^\omega$ to determine their co-lex order, where $k'=|\gamma_1|+|\gamma_2|-gcd(|\gamma_1|,|\gamma_2|)$. Our claim easily follows since $|\beta_1'\alpha_2'| = |\alpha_2'\beta_1'| = |\alpha_2|$.
\end{proof}

\begin{corollary}
The co-lex order of the infimum and supremum strings of a DFA is the same as the co-lex order of their length-$(2n)$ suffixes.
\label{cor:sortlength}
\end{corollary}
\begin{proof}
By Lemma \ref{lemma:infsup:string:form}, we can represent two infimum/supremum strings $\alpha,\beta$ as $\alpha=\alpha_2^\omega\alpha_1$ and $\beta=\beta_2^\omega\beta_1$. By the same lemma, each of $|\alpha_1|$, $|\alpha_2|$, $|\beta_1|$ and $|\beta_2|$, as well as  $|\alpha_1|+|\alpha_2|$ and $|\beta_1|+|\beta_2|$, are bounded by $n$.
From Lemma \ref{lemma:suffix:length}, it is sufficient to compare the suffixes of length at most $|\alpha_2|+|\beta_2|+\max\{|\alpha_1|,|\beta_1|\}=\max\{(|\alpha_1|+|\alpha_2|)+|\beta_2|,|\alpha_2|+(|\beta_1|+|\beta_2|)\})$ of $\alpha$ and $\beta$ in order to discover their co-lex order.
Therefore, $2n$ is a sufficient suffix length for sorting all the infimum strings correctly. 
\end{proof}

Putting everything together, we conclude:

\begin{lemma}
The infimum and supremum strings of an input-consistent DFA $\mathcal{A}=(Q,\Sigma,\delta,s,F)$ can be computed and sorted in $O(mn)$ time, where $n=|Q|$ is the number of states and $m=|\delta|$ is the number of transitions.
\end{lemma}

Equivalently, the above lemma shows that the representation of $<_{\mathcal A}$ of Section \ref{sec:interval representation} can be computed in $O(mn)$ time. Plugging the linear-time chain partition algorithm of Section \ref{sec:chain partition}, we obtain:

\begin{theorem}
Given an input-consistent DFA $\mathcal{A}=(Q,\Sigma,\delta,s,F)$, we can compute a minimum-size chain partition of $<_{\mathcal A}$ in $O(mn)$ time, where $n=|Q|$ is the number of states and $m=|\delta|$ is the number of transitions.
\end{theorem}

\subsection{\texorpdfstring{$O(n^2\log n)$}{O(n\^2 log n)}-time suffix doubling algorithm}
\label{sec:sorting:doubling}

Instead of increasing the length of the sorted suffixes only by 1 at every iteration, we can double it via a generalization of the prefix doubling algorithm \cite{manber:soda90}, the first suffix array construction algorithm that appeared in the literature. 
Again, for simplicity we describe the algorithm just for infimum strings; it is easy to modify it so that it computes and sorts the union of all infimum and supremum strings. 

Algorithm \ref{alg:suffixdoubling} describes our sorting procedure, which we explain in detail in the rest of the section. 
At every iteration $k\ge 0$, Algorithm \ref{alg:suffixdoubling} keeps the infimum strings sorted by their length-$2^k$ suffixes.
Suppose we already sorted the infimum strings with respect to their length-$2^k$ suffixes.
To enable the doubling procedure, we need to show how to compute the length-$2^{k+1}$ suffix of each $\inf I_u$ given as input the length-$2^k$ suffixes of each $\inf I_u$, for all $u\in Q$. 
Given that the infimum strings are sorted by their length-$2^k$ suffixes, for each $u\in Q$ we can achieve this goal by finding a state $v\in Q$ such that 
\[
suf_{2^{k+1}}(\inf I_u) = suf_{2^{k}}(\inf I_v) \cdot suf_{2^{k}}(\inf I_u).
\]
We call such $v$ an \emph{extender of $u$ (at distance $2^k$)}. More formally, the set $P_k(u)$ of all extenders of $u$ at distance $2^k$ is defined as: 
$$
\begin{array}{l}
P_0(u) = \{v\in\delta^{-1}(u): (\forall v'\in \delta^{-1}(u))(\lambda(v)\le\lambda(v')) \}\\
P_k(u)=\{v\in Q: \delta(v,suf_{2^{k}}(\inf I_u))=u\ \wedge\ suf_{2^{k}}(\inf I_v)\sqsubseteq suf_{2^{k+1}}(\inf I_u)\} \ \ for\ k>0
\end{array}
$$
For $u\in Q$, let $rank_{2^k}(u)$ be the co-lex rank of $suf_{2^k}(\inf I_u)$, as defined in the previous section. 
Observe that, by definition, for every $u\in Q$ and $v_1,v_2\in P_k(u)$, $rank_{2^k}(v_1)=rank_{2^k}(v_2)$. 

\begin{algorithm}[t]
\caption{Suffix doubling algorithm for sorting the infimum strings of a DFA}\label{alg:suffixdoubling}
  \KwInput{An input-consistent DFA $\mathcal{A} = (Q,\Sigma,\delta,s,F)$}
  \KwOutput{$rank_{2^k}(u)$ for each $u\in Q$, with $2n \leq 2^k < 4n$.}
  $k\gets 0$\;
  \For{$u\in Q$}{
     $rank_{2^k}(u)\gets\lambda(u)$\; $P_k(u)\gets\{ v\in\delta^{-1}(u)\ :\ \big(\forall v'\in\delta^{-1}(u)\big)\big(\lambda(v)\le\lambda(v')\big)\}$\;
    }
  \While{$2^k<2n$}
  {
    \For{$u\in Q$}{
    $a_u \gets rank_{2^k}(u)$\;
    \eIf{$P_k(u) = \emptyset$}{
        $b_u \gets -\infty$\;
    }{
        Pick any $v\in P_k(u)$\;
        $b_u \gets rank_{2^k}(v)$\; 
    }
    }
    Compute $rank_{2^{k+1}}(\cdot)$ by radix-sorting pairs $(a_u,b_u)$\;
    \For{$u\in Q$}{
        $\hat P_{k+1}(u)\gets\bigcup_{v\in P_k(u)} P_k(v)$\;
        $P_{k+1}(u)\gets\{ v\in \hat P_{k+1}(u) : \big(\forall v'\in \hat P_{k+1}(u) \big)\big(rank_{2^{k+1}}(v)\le rank_{2^{k+1}}(v')\big) \}$\;
    }
    $k\gets k+1$\;
  }
\Return{$rank_{2^k}(\cdot)$} 
\end{algorithm}

We implement a suffix doubling step as follows. Assume $P_k(u)$ and $rank_{2^k}(u)$ have been computed for all $u\in Q$. We associate with $u$ the pair $(a_u,b_u)$ where $a_u=rank_{2^k}(u)$ and $b_u$ is chosen as follows. If $P_k(u)\ne \emptyset$, $b_u=rank_{2^k}(v)$ with any $v\in P_k(u)$; otherwise, $b_u=-\infty$ is chosen\footnote{In this case, there are no more characters to be prepended to $\inf I_u$ (i.e. $|\inf I_u| < 2^k$). Since we radix-sort pairs $(a_u,b_u)$, this choice is consistent with the fact that $suf_{2^k}(\inf I_u)$ is left-padded with copies of symbol $\#$ in order to reach length $2^k$, with $\#<c$ for all $c\in \Sigma$ (see definition of $suf_{2^k}$ in the previous section).}. 
Finally, we compute $rank_{2^{k+1}}(\cdot)$ by radix-sorting pairs $(a_u,b_u)$ in $O(n)$ time.

After computing $rank_{2^{k+1}}(\cdot)$, we need to compute $P_{k+1}(\cdot)$ for the next doubling step.
For a state $u\in Q$, let $
\hat{P}_{k+1}(u)=\bigcup_{u'\in P_{k}(u)}P_k(u')$ be the union of the extender sets of $u$'s extenders at distance $2^{k}$. 
Then, we claim that we can compute $P_{k+1}(u)$ by removing all non-minimum-rank states (i.e. non-minimum $rank_{2^{k+1}}(\cdot)$) from $\hat{P}_{k+1}(u)$.
The correctness of this procedure follows from the fact that $P_{k+1}(u)$ can also be defined as the largest subset of $\hat{P}_{k+1}(u)$ 
such that, for every $v\in P_{k+1}(u)$ and $\hat v\in \hat{P}_{k+1}(u)$, $rank_{2^{k+1}}(v)\le rank_{2^{k+1}}(\hat v)$.  
To see this, first observe that $P_{k+1}(u)\subseteq\hat{P}_{k+1}(u)$ because (i) $v\in \hat{P}_{k+1}(u)$ if and only if $\delta(v,suf_{2^{k+1}}(\inf I_u))=u$, and (ii) $v\in P_{k+1}(u) \Rightarrow \delta(v,suf_{2^{k+1}}(\inf I_u))=u$.
Also, by the definition of $P_{k+1}$, $suf_{2^{k+1}}(\inf I_v)$ for $v\in P_{k+1}(u)$ must not be greater than $suf_{2^{k+1}}(\inf I_{\hat{v}})$ for any $\hat v\in \hat{P}_{k+1}(u)$, which is equivalent to $rank_{2^{k+1}}(v)\le rank_{2^{k+1}}(\hat v)$; otherwise, $ suf_{2^{k+1}}(\inf I_{\hat{v}}) \cdot suf_{2^{k+1}}(\inf I_u)< suf_{2^{k+1}}(\inf I_v)\cdot suf_{2^{k+1}}(\inf I_u)=suf_{2^{k+2}}(\inf I_u)$, which contradicts the definition of $\inf I_u$.

Since $rank_{2^{k+1}}(v)$ has already been computed for all $v\in Q$ and can thus be evaluated in constant time, from the above  characterization of  $P_{k+1}(u)$ we obtain that the time required to compute this set is proportional to the time we spend to compute the union $\hat{P}_{k+1}(u)=\bigcup_{u'\in P_{k}(u)}P_k(u')$. Observe that, if there were repeated states among the sets $P_k(u')$, for $u'\in P_{k}(u)$, then computing such a union could take time $O(n^2)$ (for every $u\in Q$), leading to a cubic algorithm. Luckily, with the next lemma we show that this is not the case: being the input automaton deterministic, those sets are pairwise disjoint and their union can thus be computed by just concatenating them. 

\begin{lemma}
Let $u\in Q$ be a state of a DFA, and let $v_1,v_2(\neq v_1)\in P_k(u)$ be extenders of $u$ at distance $2^k$. Then $P_k(v_1)\cap P_k(v_2)=\emptyset$.
\label{lem:update:pred}
\end{lemma}
\begin{proof}
Let $v_1,v_2 (\neq v_1)\in P_k(u)$ be extenders of the same state $u\in Q$ at distance $2^k$.
Let $\alpha_1=suf_{2^k}(\inf I_{v_1})$ and $\alpha_2=suf_{2^k}(\inf I_{v_2})$.
Assume, for a contradiction, that there exists $v'\in P_k(v_1)\cap P_k(v_2)$.
By definition of $P_k$, since $v'\in P_k(v_1)$, it holds $\delta(v',\alpha_1)=v_1$.
Similarly, it also holds $\delta(v',\alpha_2)=v_2$.
Since $v_1,v_2\in P_k(u)$ are extenders of the same state $u\in Q$ at distance $2^k$, both $\alpha_1$ and $\alpha_2$ are equal to the length-$2^k$ prefix of $suf_{2^{k+1}}(\inf I_u)$, therefore $\alpha_1=\alpha_2$.
Consequently, we have $v_1=\delta(v',\alpha_1)=\delta(v',\alpha_2)=v_2$, i.e. reading a string $\alpha_1=\alpha_2$ from a state $v'$ we reach two distinct states $v_1\neq v_2$. This is a contradiction with the fact that the automaton is deterministic, so the claim $P_k(v_1)\cap P_k(v_2) = \emptyset$ must be true. 
\end{proof}

From Lemma \ref{lem:update:pred}, we can compute $\hat P_{k+1}(u)$ in time proportional to its cardinality $|\hat P_{k+1}(u)|\le n$. Since finding the minimum-rank states can be done in linear $O(|\hat P_{k+1}(u)|)$ time as well, the computation of $P_{k+1}(u)$ takes time $O(n)$ for each $u\in Q$.
We conclude that each iteration of the suffix doubling algorithm takes $O(n^2)$ time.

It is worth noting that, since we keep track of each set $P_k(u)$, 
the running time of an iteration is lower-bounded by the total number of extenders therein. 
In the worst case, however, the total number of extenders at a single iteration could be truly quadratic even on acyclic DFAs: see Figure \ref{fig:num:pred}; in this example, there are $n=4\sigma+2$ states and $\sigma^2+2\sigma+1=\Theta(n^2)$ extenders at distance $2^k$ for $k=1$.

\begin{figure}[t]
\centering
\includegraphics[page=1]{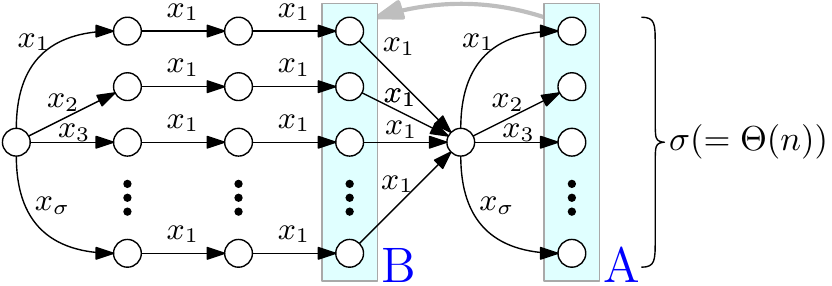}
\caption{The DFA that has a quadratic number of extenders: the $\sigma = \Theta(n)$ states in the rightmost column (indicated with A) have $\sigma  = \Theta(n)$ extenders each (indicated with B) at distance $2^k$, where $k=1$.}
\label{fig:num:pred}
\end{figure}

Putting all together, we have the following result for the suffix doubling algorithm described in this section.

\begin{lemma}
The infimum and supremum strings of an input-consistent DFA $\mathcal{A}=(Q,\Sigma,\delta,s,F)$ can be computed and sorted in $O(n^2\log n)$ time, where $n=|Q|$ is the number of states.
\end{lemma}

Equivalently, the above lemma shows that the representation of $<_{\mathcal A}$ of Section \ref{sec:interval representation} can be computed in $O(n^2\log n)$ time. Plugging the linear-time chain partition algorithm of Section \ref{sec:chain partition}, we obtain:

\begin{theorem}
Given an input-consistent DFA $\mathcal{A}=(Q,\Sigma,\delta,s,F)$, we can compute a minimum-size chain partition of $<_{\mathcal A}$ in $O(n^2\log n)$ time, where $n=|Q|$ is the number of states.
\end{theorem}

\noindent 
{\textbf{The suffix doubling algorithm in practice}}
Although every iteration of the suffix doubling algorithm needs to keep track of $O(n^2)$ extenders per iteration in the worst case, we conjecture that it is not likely to have a quadratic number of extenders on realistic datasets. To demonstrate this, we conducted a brief experiment using a pan-genomic graph, which is considered to be one of the most important real-world applications of our problem. We downloaded the Chromosome 22 sequence of the GRCh38 human reference genome and its variation data from 1000 Genome project \cite{data:pangenome}. This variation dataset contains a set of substitutions, insertions and deletions appearing on the reference human genome sequence collected from 2,548 samples.
Using this dataset, we constructed a pan-genomic graph using \textsf{VG} \cite{siren19bioinformatics}, then converted it into a DFA using the classical powerset construction algorithm \cite{nfadfa:powerset}.
We ran an implementation of our suffix doubling algorithm to sort the infimum and supremum strings and measured the number of extenders at each iteration.
The largest $\hat P_k(u)$ and $P_k(u)$ (extenders at distance $2^k$ before/after filtering non-minimum-rank states) during the procedure had cardinality 60 and 37, respectively, which might be considered not negligible but quite small when compared to the DFA's size ($n$=51,904,782, $m=$53,049,316). 
In addition, the sum $\sum_{u\in Q} |\hat P_k(u)|$ of the number of extender candidates (the union of extenders before filtering non-minimum-rank states) at any fixed distance $2^k$ was at 
most two times the number of edges, suggesting that in practice our algorithm exhibits a linearithmic complexity on pan-genomic graphs. 
C++ source code is avaliable at: \url{https://github.com/regindex/DFA-suffix-doubling}.

\subsection{\texorpdfstring{$O(m\log n)$}{O(m log n)}-time algorithm for acyclic DFAs}

If the input DFA is acyclic, then we can sort the infimum strings more efficiently using the algorithm described in \cite[Sec. 3.2]{alanko:soda20}.
This algorithm processes the states of any acyclic \emph{Wheeler DFA} $\mathcal A$ (that is, $width(\mathcal A)=1$) and their incoming edges in any topological order $u_1, \dots, u_n$ while updating $<_{\mathcal A}$ in an online fashion; more precisely, as soon as step $1\leq i \leq n$ has finished, the algorithm has computed the total order $<_{\mathcal A}$ of the set $\{u_1, \dots, u_i\}$.
The basic idea is to process the states in any topological order while maintaining a dynamic data structure that stores the relative co-lex order of the states according to any representative of $I_u$ (in fact, \cite{alanko:soda20} proves that on Wheeler DFAs, any  string in $I_u$ can be chosen as a representative of the whole $I_u$ to sort the automaton's states).
This is possible because, when $u_i\in Q$ is being processed, the structure is able to check if $rank(v)\le rank(u_j)$ for any $v\in\delta^{-1}(u_i)$ and $j<i$, where $rank(v)$ denotes the position of $v$ in the total order $<_{\mathcal A}$ of the already-processed states $u_1, \dots, u_{i-1}$ (and, by definition of topological order, $rank(v)$ and $rank(u_j)$ have already been computed in the previous steps). This information is sufficient to compute $rank(u_i)$ among the sorted $u_1, \dots, u_{i}$.

In our case (arbitrary acyclic DFAs), we use the above data structure as follows: after topologically sorting $\mathcal A$ (in linear time) we process states in this order. 
When processing state $u_i$,  assume that states $u_1, \dots, u_{i-1}$ have already been co-lex sorted according to their strings $\inf I_{u_1}, \dots, \inf I_{u_{i-1}}$ using the data structure of \cite[Section 3.2]{alanko:soda20}.
By scanning the predecessors of $u_i$, we find the min-rank state $v^*=\arg\min_{v\in \delta^{-1}(u_i)} rank(v)$ among them.
At this point, we insert state $u_i$, as well as transition (labeled edge) $(v^*,u_i, \lambda(u_i))$, in the data structure. Note that, since only $(v^*,u_i, \lambda(u_i))$ is inserted, the data structure  of \cite[Section 3.2]{alanko:soda20} maintains a spanning tree of $\mathcal{A}$ rooted at the start state $s$.
By construction, it is easy to see that the unique path connecting $s$ to $u_i$ in this spanning tree is labeled with string $\inf I_{u_i}$: the spanning tree encodes the infimum strings. As a result, our sorting problem is equivalent to sorting this spanning tree.
It is known that a labeled tree is a special case of Wheeler graphs (see \cite{gagie:tcs17}), so the computed co-lex node order of this spanning tree is precisely the co-lex order of all infimum strings. 

Since it is immediate to extend the idea to the union of all infimum and supremum strings, taking into account the cost of each update of the data structure \cite[Sec. 3.2]{alanko:soda20},  we obtain: 

\begin{lemma}
The infimum and supremum strings of an input-consistent acyclic DFA $\mathcal{A}=(Q,\Sigma,\delta,s,F)$ can be computed and sorted in $O(m\log n)$ time where $n=|Q|$ is the number of states and $m=|\delta|$ is the number of transitions.
\end{lemma}
\begin{proof}
First of all, the data structure \cite[Sec. 3.2]{alanko:soda20} supports the following two operations in $O(\log n)$ time\footnote{In fact, the running time per operation is $O(\log m')$, where $m'$ is the number of edges. However, in our case, the number of transitions inserted into the data structure is only $m' = n-1$ since we insert one transition per state (except $s$). Therefore, in our case, the time taken per operation is $O(\log n)$.}:
(i) computing the relative rank of a state among those that are already processed, and (ii) inserting a new edge (a state is inserted into the structure after all its incoming edges have been inserted).
As a result, finding  $v^*=\arg\min_{v\in \delta^{-1}(u_i)} rank(v)$ takes time $O(|\delta^{-1}(u_i)|\cdot \log n)$.  
After $v^*$ has been found, inserting the labeled edge $(v^*,u_i, \lambda(u_i))$, as well as state $u_i$, into the structure takes time $O(\log n)$. 
Overall, after all states have been processed the cost of the above operations amounts to  $O(m\log n)$ time.
Since a topological order of $\mathcal A$ can be computed in $O(m)$ time, the total running time is $O(m\log n)$.
\end{proof}

Equivalently, the above lemma shows that the representation of $<_{\mathcal A}$ of Section \ref{sec:interval representation} can be computed in $O(m\log n)$ time when $\mathcal A$ is acyclic. Plugging the linear-time chain partition algorithm of Section \ref{sec:chain partition}, we obtain:

\begin{theorem}
Given an input-consistent acyclic DFA $\mathcal{A}=(Q,\Sigma,\delta,s,F)$, we can compute a minimum-size chain partition of $<_{\mathcal A}$ in $O(m\log n)$ time, where $n=|Q|$ is the number of states and $m=|\delta|$ is the number of transitions.
\end{theorem}

%%
%% Bibliography
%%

\end{document}